\documentclass[12pt]{amsart}
\usepackage{graphicx,amsmath,amssymb}

\numberwithin{equation}{section}

\newcommand{\C}{\mathbb C}
\newcommand{\R}{\mathbb R}
\newcommand{\Z}{\mathbb Z}
\newcommand{\N}{\mathbb N}

\renewcommand{\d}{\prime}
\newcommand{\dd}{{\prime \prime}}

\renewcommand{\Re}{{\rm Re}\,}
\renewcommand{\Im}{{\rm Im}\,}
\newtheorem{theorem}{Theorem}[section]
\newtheorem{lemma}[theorem]{Lemma}

\newtheorem{corollary}[theorem]{Corollary}

\newtheorem{definition}{Definition}
\newtheorem*{remark}{Remark}

\begin{document}
\title[]
{Asymptotics of eigenvalues of non-self adjoint Schr\"odinger operators on a half-line}
\author[]
{Kwang C. Shin}

\address{Department of Mathematics, University of West Georgia, Carrollton, GA USA}
\date{January 26, 2009}

\begin{abstract}
We study the eigenvalues of  the non-self adjoint problem $-y^\dd+V(x)y=E y$ on the half-line $0\leq x<+\infty$
under the Robin boundary condition at $x=0$, where $V$ is a monic polynomial of degree $\geq 3$. We obtain a
Bohr-Sommerfeld-like asymptotic formula for $E_n$ that depends on the boundary conditions. Consequently, we solve
certain inverse spectral problems, recovering the potential $V$ and boundary condition from the first $(m+2)$
terms of the asymptotic formula.
\end{abstract}

\maketitle

\begin{center}
\end{center}

\baselineskip = 18pt

\section{\bf Introduction and results}
\label{introduction} 
In this paper, we study  Schr\"odinger eigenvalue problems with
real or complex polynomial potentials on the half-line
$[0,\infty)$ under a Robin  boundary condition at $0$. We develop
a new kind of
 asymptotic formula relating the index $n$ to a series of fractional powers of the eigenvalue $E_n$
(see Theorem \ref{the_main}).  The coefficients   in the series are explicit polynomials given in terms of the
coefficients of the polynomial potential. In contrast, the asymptotic formulas in Titchmarsh \cite{TIT1} of the
Bohr-Sommerfeld type (see \eqref{asy_TIT} and  \eqref{asy_TIT1}) seem in practice to require numerical
integrations in the complex plane. Another advantage of our method is that  it is easier to improve the error
terms,  as we will explain in the remark at the end of Section \ref{main_sect}.

For an integer $m\geq 3$,  consider the Schr\"odinger eigenvalue problem
\begin{equation}\label{ptsym}
-y^\dd(x)+\left[x^m+P(x)\right]y(x)=E y(x),\quad 0\leq x<+\infty,
\end{equation}
under the Robin boundary condition at $x=0$
\begin{equation}\label{bd_cond}
 y(0)\cos\theta+
y^\d(0)\sin\theta=0\quad\text{and}\quad y(+\infty)=0,
\end{equation}
with  $\theta\in\C$ fixed, where the lower order terms in the potential are
\begin{equation}
P(x)=a_1x^{m-1}+a_2x^{m-2}+\cdots+a_{m-1}x+a_m,\quad a_j\in\C.
\end{equation}
Write $ {a}=(a_1,\,a_2,\dots,\, a_m)$ for the coefficient vector of $P$.

Sibuya \cite{Sibuya} showed that there are infinitely many eigenvalues, by showing that eigenvalues are zeros of
an entire function of order $\frac{1}{2}+\frac{1}{m}$.

Titchmarsh \cite{TIT1} showed that  under the Dirichlet boundary condition, \eqref{bd_cond} with $\sin \theta=0$,
the eigenvalue problem \eqref{ptsym} has the eigenvalues $\{E_n\}_{n\geq0}$ satisfying
\begin{equation}\label{asy_TIT}
\int_0^{X_n}\sqrt{E_n-x^m-P(x)}\,dx=\left(n+\frac{1}{4}\right)\pi+O\left(\frac{1}{n}\right),
\end{equation}
where $X_n=X(E_n)>0$ solves $X_n^m+P(X_n)=E_n$, provided that $x^m+P(x)$ and its first derivative are
non-decreasing on the interval $(0,+\infty)$. With this restricted class of non-decreasing and convex polynomial
potentials under the Neumann boundary condition $\theta=\frac{\pi}{2}$ at $x=0$, Titchmarsh  also showed that
\begin{equation}\label{asy_TIT1}
\int_0^{X_n}\sqrt{E_n-x^m-P(x)}\,dx=\left(n+\frac{3}{4}\right)\pi+O\left(\frac{1}{n}\right).
\end{equation}
 Then he claimed without proof that \eqref{asy_TIT} holds under the Robin boundary condition \eqref{bd_cond}
($\sin\theta\not=0$). That claim is wrong except for $\sin\theta=0,\,\pm1$. There must be an additional term
$\left(\cot\theta\right)E_n^{-\frac{1}{2}}=(const.)n^{-\frac{m}{m+2}}$ on the left-hand side of \eqref{asy_TIT},
as we will see in Theorem \ref{the_main}.

Let us compare these asymptotic formulas with the following main result of this paper.
\begin{theorem}\label{the_main}
If $\sin\theta\not=0$ (the Robin boundary condition at $x=0$, other than  Dirichlet), then  the eigenvalues
$\left\{E_n\right\}_{n=0}^{\infty}$ of \eqref{ptsym} with boundary condition \eqref{bd_cond} satisfy that
\begin{equation}\label{N_boundary}
\sum_{j=0}^{m+1}c_j( {a})E_n^{\frac{1}{2}+\frac{1-j}{m}}+\pi^{-1}(\cot\theta)
E_n^{-\frac{1}{2}}+O\left(E_n^{-\frac{1}{2}-\frac{1}{m}}\right)\underset{n\to+\infty}{=}\left(n+\frac{1}{4}\right),
\end{equation}
where the $c_0( {a})= (2\pi)^{-1}B\left(\frac{1}{2},\,1+\frac{1}{m}\right)$ and the $c_j( {a})$ are explicit
polynomials in the coefficients $ {a}$ of the polynomial potential, defined  in \eqref{c_def}.

If $\sin\theta=0$ (the Dirichlet boundary condition at $x=0$), then   the eigenvalues
$\left\{E_n\right\}_{n=0}^{\infty}$ of \eqref{ptsym} with boundary condition \eqref{bd_cond} satisfy that
\begin{equation}\label{D_boundary}
\sum_{j=0}^{m+1}c_j( {a})E_n^{\frac{1}{2}+\frac{1-j}{m}}
+O\left(E_n^{-\frac{1}{2}-\frac{1}{m}}\right)\underset{n\to+\infty}{=}\left(n+\frac{3}{4}\right).
\end{equation}
\end{theorem}
The formulas in Theorem \ref{the_main} are explicit and ready to be inverted, for they express $n$ in terms of
$E_n$. In contrast, the asymptotic formulas \eqref{asy_TIT} and \eqref{asy_TIT1} would need  first to be
numerically integrated and then inverted.

To the best of my knowledge, Theorem \ref{the_main} is the first
result of its type. Further, we treat arbitrary real and complex
monic polynomial potentials of degree $\geq 3$. Titchmarsh
\cite{TIT1} required his potentials to be non-decreasing and
convex, and polynomial-like at infinity.

Another motivation for Theorem \ref{the_main} is that this
half-line problem is a building block for the full line problem
and other ``connection" problems in the complex plane whose
spectral determinants can be described by the spectral
determinants of two half-line problems of the form \eqref{ptsym}
(see, e.g., \cite{Shin2, Sibuya}).

This paper is organized as follows. In Section \ref{sec2-1}, we present direct and inverse spectral results
implied by Theorem \ref{the_main}. In Section \ref{sec2}, we review a method of Sibuya \cite{Sibuya}, regarding
solutions of \eqref{ptsym} satisfying $y(+\infty)=0$. We  outline his proof of Theorem \ref{prop-2} that studies
asymptotics of certain solutions of \eqref{ptsym}. In Section \ref{main_sect}, we improve the error terms of the
solutions in Theorem \ref{prop-2} and indicate how one can improve the error terms even further.

The asymptotics of solutions in Sections \ref{sec2} and \ref{main_sect} are valid when $E$ is far away from the
positive real axis while all but finitely many eigenvalues lie near the positive real axis in the complex plane.
In Section \ref{prop_sect}, by using a connection formula we will express solutions having $E$  near the positive
real axis by a linear combination of solutions having $E$ away from the positive real axis. In Section
\ref{main_thm_sec}, we prove the main result Theorem \ref{the_main}. Finally, in the Appendix we will compute
various constants that appear in our asymptotic formulas.

The interest of physicists to one-dimensional Schr\"odinger eigenvalue problem with complex potentials renewed in
recent years due to the discovery that certain ``$\mathcal{PT}$-symmetric'' potentials have purely real spectrum
\cite{Bender, Bender2, CCG, Dorey, Levai, Shin, VOR3}.  For general background on complex differential equations,
one can consult classic monographs by Hille \cite{Hille,Hille1} or by Sibuya \cite{Sibuya}. Also, Eremenko {\em
et.\ al.}   recently studied some Schr\"odinger equations with polynomial potentials and derived interesting
results on analytic continuation of  eigenvalues \cite{Eremenko1} and on location of zeros of eigenfunctions
\cite{Eremenko2,Eremenko3}.

\section{\bf Direct and inverse spectral consequences of the main result}\label{sec2-1}
In this section, we will derive some consequences of Theorem \ref{the_main}. Fix $m\geq 3$ and $\theta\in\C$, and
let $E_n$ be the eigenvalues of \eqref{ptsym} under boundary condition \eqref{bd_cond}.

One can invert the asymptotic formulas in Theorem \ref{the_main} to obtain formulas for  $E_n$ in terms of $n$.
\begin{corollary}
 One can compute numbers $d_j( {a},\theta)$ explicitly such that
\begin{equation}\label{somm_eq}
E_n\underset{n\to+\infty}{=}\sum_{j=0}^{m+1}d_j( {a},\theta)\cdot
n^{\frac{2m}{m+2}\left(1-\frac{j}{m}\right)}+O\left(n^{-\frac{4}{m+2}}\right).
\end{equation}
\end{corollary}
\begin{proof}
Equations \eqref{N_boundary} and \eqref{D_boundary} are  asymptotic equations and they can be solved for $E_n$,
resulting in \eqref{somm_eq}. For details  see for example \cite[\S 5]{Shin2}.
\end{proof}
\begin{remark}
{\rm  When one inverts the Robin asymptotics \eqref{N_boundary}, one usually writes the result in terms of powers
of $\left(n+\frac{1}{4}\right)$. Note that using the generalized binomial expansions, these powers of
$\left(n+\frac{1}{4}\right)$ can be expanded to fit in \eqref{somm_eq} and vice versa. The reason we have
\eqref{somm_eq} in the present form is to treat both the Robin and Dirichlet asymptotics as one.  }
\end{remark}

In next corollary, we provide an asymptotic formula for nearest neighbor spacing of eigenvalues. Hence we deduce
that large eigenvalues increase monotonically in magnitude and have their argument approaching zero.
\begin{corollary}\label{monoton}
The space between successive eigenvalues is
\begin{equation}
 E_{n+1}-E_n
\underset{n\to+\infty}{=}d \cdot n^{\frac{m-2}{m+2}}+o\left(n^{\frac{m-2}{m+2}}\right),
\end{equation}
where
\begin{equation}
d=\frac{2m}{m+2}\left(\frac{2\pi}{B\left(\frac{1}{2},1+\frac{1}{m}\right)}\right)^{\frac{2m}{m+2}}.
\end{equation}
In particular, $\lim_{n\to\infty}\left|E_{n+1}-E_n\right|=\infty$ and $\lim_{n\to+\infty}\arg(E_n)=0.$ Hence:

\begin{equation}\nonumber
\left|E_n\right|<\left|E_{n+1}\right|\,\,\,\text{for all large $n$}.
\end{equation}
\end{corollary}
\begin{proof}
 These claims are consequences of \eqref{mono} and \eqref{mono2} with $E_n=-\lambda_n$ in the proof of Theorem \ref{the_main}.
Note that in that proof,
\begin{equation}
d=(c_0)^{-\frac{2m}{m+2}}\cdot{\frac{2m}{m+2}\choose 1}.
\end{equation}
\end{proof}

The next corollary shows that one can recover the polynomial potential and boundary condition from the asymptotic
formula for the eigenvalues.
\begin{corollary}\label{cor5}
   Suppose  $\sin\theta\not=0$ and
\begin{equation}\label{bohr_eq}
\sum_{j=0}^{m+1}c_j^*
E_n^{\frac{1}{2}+\frac{1-j}{m}}+O\left(E_n^{-\frac{1}{2}-\frac{1}{m}}\right)\underset{n\to+\infty}{=}\left(n+\frac{1}{4}\right)
\end{equation}
 for some $c_j^*\in\C$.  Then one can recover the polynomial potential and boundary condition.
\end{corollary}

\begin{proof}
The coefficients $c_j( {a})$ in Theorem \ref{the_main} have the following properties (see Remark at the end of
Section \ref{main_thm_sec}).
\begin{itemize}
\item[]{(i)} The $c_j( {a})$ are all polynomials in terms of the coefficients $ {a}$ of $P(x)$.
\item[]{(ii)} The coefficient $c_0( {a})$ does not depend on $ {a}$ (it is a constant).
\item[]{(iii)} For $1\leq j\leq m$, the polynomial $c_j( {a})$ depends only on   $a_1,\,a_2,\dots,a_{j}$. Further, it is a non-constant linear function of $a_j$.
\end{itemize}
We will use induction on $j$. Since $c_1( {a})$ is a  non-constant linear function of $a_1$, $c_1( {a})=c_1^*$
determines $a_1$.  Suppose that for $1\leq j\leq m-1$, the $c_1^*,  c_2^*,\dots,c_j^*$ uniquely determine $a_1,\,
a_2,\dots, a_j$. Then by (iii), $c_{j+1}( {a})=c_{j+1}^*$ determines $a_{j+1}$ uniquely.  These coefficients $
{a}$ of $P(x)$ together with the $c_{m+1}^*$-term then determine $\cot (\theta)$, which determines the boundary
condition \eqref{bd_cond}.
\end{proof}

\begin{remark}
{\rm Statements similar to the last Corollary hold when $\sin\theta=0$. That is, the first $(m+2)$-terms of the
asymptotics and boundary condition will determine the polynomial potential.}
\end{remark}

\begin{corollary}
The following statements are equivalent:
\begin{itemize}
\item[]{(i)} There are infinitely many real eigenvalues.
\item[]{(ii)} All eigenvalues are real.
\item[]{(iii)} The polynomial $P$ is real-valued and $\theta\in\R$.
\end{itemize}
\end{corollary}
\begin{proof}
$(i)\Rightarrow (iii):$ We will prove its contrapositive.  Suppose that either $P$ is not real-valued or
$\theta\not\in\R$. Then from Theorem \ref{the_main}, one can see that all but finitely many eigenvalues are
non-real. Thus, there are at most finitely many real eigenvalues. $(iii)\Rightarrow (ii):$ The differential
operator is self-adjoint. $(ii) \Rightarrow (i):$ There are infinitely many eigenvalues.
\end{proof}

\section{\bf Work of Sibuya}\label{sec2}
In this section, we will introduce some results of Sibuya and outline his proof of a result that we will use to
improve certain error terms.

Sibuya \cite{Sibuya} studied the following equation in the complex
$z$-plane.
\begin{align}\label{complex-1}
-y^{\dd}+\left(z^m+P(z)+\lambda\right)y=0,
\end{align}
that is \eqref{ptsym}, extended to the complex $z$-plane with $\lambda=-E$. Before we introduce the Sibuya's
results, we provide some notations first. We will use $ {a}=(a_1, a_2,\dots,a_m)\in\C^{m}$,
\begin{equation}
\text{$b_{j,k}( {a})$ is the coefficient of $z^{mk-j}$ in ${\frac{1}{2}\choose{k}}\left(P(z)\right)^k$\,\,for
$1\leq k\leq j$,\, and}
\end{equation}
\begin{equation}
b_j( {a})=\sum_{k=1}^j b_{j,k}( {a}),\quad j\in\N.
\end{equation}
We will also use $\mu( {a})=\frac{m}{4}-\nu( {a})$ where
\begin{eqnarray}
\nu( {a})=\left\{
              \begin{array}{rl}
              0 \quad &\text{if $m$ is odd,}\\
              b_{\frac{m}{2}+1}( {a}) \quad &\text{if $m$ is even.}
              \end{array}
                         \right. \nonumber
\end{eqnarray}

Sibuya's results that we need in this paper can be summarized as
follows.
\begin{theorem}\label{prop-2}
Equation \eqref{complex-1} admits a solution $f(z,  {a}, \lambda)$ with the following properties.
\begin{enumerate}
\item[(i)] $f(z,  {a}, \lambda)$ is an entire function of $z$, $a$,
and $\lambda$. \item[(ii)] $f(z,  {a}, \lambda)$ and $f^\d(z,  {a}, \lambda)=\frac{\partial}{\partial z}f(z,
 {a}, \lambda)$ admit the following asymptotic expansions:
\begin{align}
f(z, {a},\lambda)=&\,z^{\mu( {a})-\frac{m}{2}}(1+O(z^{-1/2}))\exp\left[-F(z, {a},\lambda)\right],\label{Si-eq1}\\
f^\d(z, {a},\lambda)=&-z^{\mu( {a})}(1+O(z^{-1/2}))\exp\left[-F(z, {a},\lambda) \right],\label{Si-eq2}
\end{align}
as $z$ tends to infinity in the sector $|\arg z|\leq \frac{3\pi}{m+2}-\varepsilon$, for any $\varepsilon>0$ fixed,
uniformly on each compact set of $( {a},\lambda)$-values. Here
\begin{equation}\nonumber
F(z, {a},\lambda)=\frac{2}{m+2}\,z^{\frac{m+2}{2}}+\sum_{1\leq j<\frac{m+2}{2}}\frac{2}{m+2-2j}b_j( {a})
z^{\frac{1}{2}(m+2-2j)}.
\end{equation}
\item[(iii)] For each fixed $\delta>0$, $f$ and $f^\d$ also admit
the asymptotic expansions,
\begin{align}
f(0,  {a}, \lambda)=&\quad [1+o(1)]\lambda^{-\frac{1}{4}}\exp\left[L( {a}, \lambda)\right],\label{Si-eq3}\\
f^\d(0,  {a}, \lambda)=&-[1+o(1)]\lambda^{\frac{1}{4}}\exp\left[L( {a}, \lambda)\right],\label{Si-eq4}
\end{align}
 as $\lambda\to\infty$ in the sector
 $|\arg(\lambda)|\leq\pi-\delta$,
 where
\begin{equation}\label{L_def}
L( {a},\lambda)=\int_0^{+\infty}\left(\sqrt{x^m+P(x)+\lambda}-
x^{\frac{m}{2}}-\sum_{j=1}^{\lfloor\frac{m+1}{2}\rfloor}b_j( {a})x^{\frac{m}{2}-j}-\frac{\nu(
{a})}{x+1}\right)\,dx
\end{equation}
\item[(iv)] The entire functions $f(0,  {a}, \lambda)$ and $f^\d(0,
 {a}, \lambda)$ have orders $\rho:=\frac{1}{2}+\frac{1}{m}$.
\end{enumerate}
\end{theorem}

\begin{remark}{\rm The Dirichlet and Neumann eigenvalues $E=-\lambda$ satisfy
$f(0, {a},\lambda)=0$ and $f^\d(0, {a},\lambda)=0$, respectively since $f(+\infty, {a},\lambda)=0$ for all
$\lambda$. In Lemma \ref{asy_lemma}, we will show that $L( {a},\lambda)=(const.)\lambda^{\rho}(1+o(1))$ as
$\lambda\to\infty$ in the sector $|\arg(\lambda)|\leq\pi-\delta$. Thus, there are at most finitely many such
eigenvalues in the sector. Since the orders of entire functions $\lambda\mapsto f(0, {a},\lambda)$ and
$\lambda\mapsto f^\d(0, {a},\lambda)$ are $\rho\in(0, 1)$, by the Hadamard factorization theorem, there are
infinitely many zeros of these functions  which lie in the sector $|\arg(\lambda)-\pi|\leq \delta$ (near the
negative real $\lambda$-axis) with at most finitely many exceptions.

Then using asymptotics of $f(0, {a},\lambda)$ and $f^\d(0, {a},\lambda)$ away from the negative real axis and a
functional equation \eqref{basic_eq}, we will get asymptotics of these functions near the negative real
$\lambda$-axis and use these asymptotics near the negative real $\lambda$-axis to eventually prove the main
result, Theorem \ref{the_main}.}
\end{remark}

We will improve these asymptotics of $f(0, {a},\lambda)$ and $f^\d(0, {a},\lambda)$ in Theorem \ref{prop-2}(iii),
but first we give an outline of Sibuya's proof \cite[Sect.\ 19]{Sibuya}.
\begin{proof}[Sibuya's proof of Theorem ~\ref{prop-2} (iii)]
Suppose that $u(z)=u(z, {a},\lambda)$ solves \eqref{complex-1}, that is,
\begin{equation}\label{sys}
-u^\dd+Q(z)u=0,\,\,\text{where $Q(z)=z^m+P(z)+\lambda$.}
\end{equation}
Or equivalently,
\begin{equation}\label{ch-vari}
\left(
\begin{matrix}
u(z)\\
u^\d(z)
\end{matrix}
\right)^\d = \left(
\begin{matrix}
0&1\\
Q(z)&0
\end{matrix}
\right) \left(
\begin{matrix}
u(z)\\
u^\d(z)
\end{matrix}
\right).
\end{equation}

Next, one defines functions $w_1,\,w_2$ through
\begin{equation}\label{ch-vari1}
\left(
\begin{matrix}
u(z)\\
u^\d(z)
\end{matrix}
\right) = \exp\left[-\int_0^zQ(t)^{\frac{1}{2}}dt\right]\left(
\begin{matrix}
Q(z)^{-\frac{1}{4}}&Q(z)^{-\frac{1}{4}}\\
Q(z)^{\frac{1}{4}}&-Q(z)^{\frac{1}{4}}
\end{matrix}
\right) \left(
\begin{matrix}
1&-g(z)\\
g(z) &1
\end{matrix}
\right) \left(
\begin{matrix}
w_1(z)\\
w_2(z)
\end{matrix}
\right).
\end{equation}
Let $z=x+0i$,\, $x\geq0$. Then from \eqref{ch-vari} and
\eqref{ch-vari1}, one gets
\begin{align}
w_1^\d(x,\lambda)&=2Q(x)^{\frac{1}{2}}w_1(x,\lambda) +s_{11}(x,\lambda)w_1(x,\lambda)+s_{12}(x,\lambda)w_2(x,\lambda),\label{trans-1}\\
w_2^\d(x,\lambda)&=\qquad\qquad\quad\quad\quad\,\,\,\,\
s_{21}(x,\lambda)w_1(x,\lambda)+s_{22}(x,\lambda)w_2(x,\lambda),\label{trans-2}
\end{align}
where
\begin{equation}\label{sij_def}
\left(
\begin{matrix}
s_{11}&s_{12}\\
s_{21}&s_{22}
\end{matrix}
\right)= \frac{1}{1+g^2} \left(
\begin{matrix}
hg-g^\d g&-hg^2+g^\d\\
-hg^2-g^\d&-hg-g^\d g
\end{matrix}
\right),\,\,
g=\frac{Q^\d(x)}{8Q^{\frac{3}{2}}(x)},\,\,\text{and}\,\,
h=\frac{Q^\d(x)}{4Q(x)}.
\end{equation}

 Next we consider the following integral equations transformed
from \eqref{trans-1} and \eqref{trans-2}:
\begin{align}
w_1(x,\lambda)&=-\int_x^{+\infty}\left[s_{11}(t,\lambda)w_1(t,\lambda)+s_{12}(t,\lambda)w_2(t,\lambda)\right]\exp\left[2\int_t^xQ(\sigma)^{1/2}d\sigma\right] dt,\label{int-1}\\
w_2(x,\lambda)&=1-\int_x^{+\infty}\left[s_{21}(t,\lambda)w_1(t,\lambda)+s_{22}(t,\lambda)w_2(t,\lambda)\right]\,dt.\label{int-2}
\end{align}
Sibuya \cite[Sect.\ 19]{Sibuya} showed first that for a given $0<\delta\leq\frac{\pi}{m+2}$ and for a compact set
of the coefficient vector $a\in\C^{m}$
 of $P(x)$, there exist positive numbers $\eta,\,M_0$ such that if $|\lambda|\geq M_0$
 and $|\arg(\lambda)|\leq\pi-\delta$,
 then $\Re\sqrt{Q(x)}=\Re\sqrt{x^m+P(x)+\lambda}\geq \eta\sqrt{|\lambda|}$ for  all $x\geq0$.
In order to show this, he examined $x^m+P(x)$ while the coefficient vector $a$ of $P(x)$ lies in a compact set of
$\C^{m}$. For a given $0<\delta\leq\frac{\pi}{m+2}$, there is an $x_0>0$ such that if $x\geq x_0$, then
$\left|\arg\left(x^m+P(x)\right)\right|=\left|\arg\left(1+\frac{a_1}{x}+\dots+\frac{a_{m}}{x^{m}}\right)\right|
\leq\frac{\delta}{2}$. Next, $x^m+P(x)$ for all $0\leq x\leq x_0$ lie in a disk, centered at $0$. Thus, there is
an $M_0$ such that if $|\lambda|\geq M_0$ and $\left|\arg(\lambda)\right|\leq\pi-\delta$, then
$\left|\arg(x^m+P(x)+\lambda)\right|\leq\pi-\frac{\delta}{2}$ and $\left|x^m+P(x)+\lambda\right| \geq
\frac{1}{2}\sin(\delta/2)|\lambda|$ for all $x\geq 0$ (see Fig.\ 19.2 in \cite{Sibuya}). Hence, there is an
$\eta>0$ such that if $|\lambda|\geq M_0$ and $|\arg(\lambda)|\leq\pi-\delta$,
 then $\Re\sqrt{x^m+P(x)+\lambda}\geq \eta\sqrt{|\lambda|}$ for all $x\geq0$.

Then Sibuya \cite[Lemma 19.1]{Sibuya} proved that among other things, there are solutions $W_1(x,\lambda)$ and
$W_2(x,\lambda)$ of \eqref{int-1} and \eqref{int-2} such that
\begin{equation}\label{asy_res}
W_1(x,\lambda)=o(1),\,W_2(x,\lambda)=1+o(1),
\end{equation}
\begin{itemize}
\item [(i)] as $\lambda\to\infty$ in the sector
$|\arg(\lambda)|\leq \pi-\delta$ uniformly for $0\leq x<+\infty$ and
\item [(ii)] as $x\to+\infty$ with $x\in\R$ uniformly for all large
enough $\lambda$ in the sector $|\arg(\lambda)|\leq \pi-\delta$.
\end{itemize}
To show \eqref{asy_res} Sibuya used the fact that
 $\int_0^{+\infty}\left|s_{jk}(x,\lambda)\right|\,dx=o(1)$ as $\lambda\to\infty$ in the sector $|\arg(\lambda)|\leq
\pi-\delta$. We will improve these estimates in Lemma \ref{estim}.

Then \eqref{asy_res} along with
\begin{equation}\label{asy_w12_eq}
\left(
\begin{matrix}
U(0, {a},\lambda)\\
U^\d(0, {a},\lambda)
\end{matrix}
\right) := \left(
\begin{matrix}
\lambda^{-1/4}&\lambda^{-1/4}\\
\lambda^{1/4}&-\lambda^{1/4}
\end{matrix}
\right) \left(
\begin{matrix}
1&-g(0,\lambda)\\
g(0,\lambda) &1
\end{matrix}
\right) \left(
\begin{matrix}
W_1(0, {a},\lambda)\\
W_2(0, {a},\lambda)
\end{matrix}
\right),
\end{equation}
yields the following asymptotics of $U(0, {a},\lambda)$ and $U^\d(0, {a},\lambda)$:
\begin{align}
U(0, {a},\lambda)=&\,\,\,\,\,\lambda^{-\frac{1}{4}}\left[1+o(1)\right],\label{eq-11}\\
U^\d(0, {a},\lambda)=&-\lambda^{\frac{1}{4}}\left[1+o(1)\right],\label{eq-21}
\end{align}
as $\lambda\to\infty$ in the sector $|\arg(\lambda)|\leq\pi-\delta$. Since both $f$ and $U$ are solutions of the
same equation \eqref{complex-1} and since they both decay to zero as $x\to+\infty$, they are linearly dependent
\cite[\S 7.4]{Hille}. Moreover, Sibuya \cite[Sect.\ 19]{Sibuya} showed that
\begin{equation}\label{fu_rel}
f(z, {a},\lambda)=\exp\left[L( {a},\lambda)\right]U(z, {a},\lambda),
\end{equation}
where $L( {a},\lambda)$ is defined by \eqref{L_def}.  Then \eqref{fu_rel}, \eqref{eq-11}, and \eqref{eq-21}  prove
Theorem \ref{prop-2}(iii).
\end{proof}


\section{\bf Improving Theorem \ref{prop-2}(iii)}\label{main_sect}
In this section, we will prove the following theorem, improving Theorem \ref{prop-2}(iii).
\begin{theorem}\label{improv}
Let $0<\delta\leq\frac{\pi}{m+2}$. Then
\begin{align}
f(0, {a},\lambda)=&\,\,\,\,\,\left[1+O\left(\lambda^{-\rho}\right)\right]\lambda^{-\frac{1}{4}}\exp\left[L( {a},\lambda)\right],\label{eq-1}\\
f^\d(0, {a},\lambda)=&-\left[1+O\left(\lambda^{-\rho}\right)\right]\lambda^{\frac{1}{4}}\exp\left[L(
{a},\lambda)\right],\label{eq-2}
\end{align}
 as $\lambda\to\infty$ in the sector $|\arg(\lambda)|\leq\pi-\delta$, uniformly on
each compact set of $a\in\C^{m}$.
\end{theorem}

\begin{proof}
We will prove this by improving errors $o(1)$ in \eqref{asy_res} to $O\left(\lambda^{-\rho}\right)$. In doing so, we will use better
estimates on integrations of $s_{ij}$ in \eqref{sij_def}   over the interval $[0,+\infty)$.

We first examine the integrands. Let $\gamma:=\arg(\lambda)\in(-\pi+\delta,\,\pi-\delta)$. Then for $0\leq
x<+\infty$,
\begin{equation}\nonumber
|x^m+\lambda|^2=x^{2m}+|\lambda|^2+2|\lambda|x^m\cos\gamma.
\end{equation}
So if $|\gamma|\leq\frac{\pi}{2}$, then $|x^m+\lambda|^2\geq
x^{2m}+|\lambda|^2$. And if $\frac{\pi}{2}<|\gamma|\leq\pi-\delta$,
then $\cos\gamma<0$ and
\begin{align}
|x^m+\lambda|^2&=x^{2m}+|\lambda|^2+2|\lambda|x^m\cos\gamma\nonumber\\
&\geq (1+\cos\gamma)\left(|\lambda|^2+x^{2m}\right)\nonumber\\
&\geq
\cos^2\left(\frac{\pi-\delta}{2}\right)\left(|\lambda|+x^{m}\right)^2,\nonumber
\end{align}
where we used $1+\cos\gamma=2\cos^2\left(\frac{\gamma}{2}\right)$
and $\alpha^2+\beta^2\geq\frac{1}{2}(\alpha+\beta)^2$ for
$\alpha,\beta\geq0$.

Also, since the coefficient vector $a$ lies in a fixed compact set, for all large enough $\lambda$ in the sector
$|\arg(\lambda)|\leq \pi-\delta$,
\begin{align}
\left|\frac{a_1x^{m-1}+\dots+a_m}{x^m+\lambda}\right|\leq&\frac{1}{2},&\nonumber\\
|Q(x)|\geq|x^m+\lambda|\left(1-\left|\frac{a_1x^{m-1}+\dots+a_{m}}{x^m+\lambda}\right|\right)&\geq\frac{1}{2}|x^m+\lambda|
\quad\text{for all $x\geq0$.}\nonumber
\end{align}
And for some $A_1,\,A_2>0$,
\begin{align}
|g(x)|&=\left|\frac{Q^\d(x)}{8Q^{3/2}(x)}\right|\leq A_1\frac{x^{m-1}+1}{\left(x^m+|\lambda|\right)^{\frac{3}{2}}},\nonumber\\
|h(x)|&=\left|\frac{Q^\d(x)}{4Q(x)}\right|\leq
A_2\frac{x^{m-1}+1}{x^m+|\lambda|}.
\end{align}

Moreover, $\sup_{x\geq0}|g(x,\lambda)|\to0$ as $\lambda\to\infty$ in the sector $|\arg(\lambda)|\leq \pi-\delta$
and use this to estimate $\left|1+g^2\right|^{-1}\leq2$ for all large enough $\lambda$ in the sector
$|\arg(\lambda)|\leq \pi-\delta$. So
\begin{equation}
\int_0^{+\infty}\left|h(x,\lambda)g(x,\lambda)\right|\,dx=O\left(\int_0^{+\infty}\left|\frac{x^{2m-2}}{\left(x^m+|\lambda|\right)^{\frac{5}{2}}}\right|\,dx\right)=O\left(\lambda^{-\rho}\right),
\end{equation}
substituting $x=|\lambda|^{\frac{1}{m}}t$ in the last step. We can apply the same argument  to $g^\d, hg^2, g^\d
g$ to prove the following. On any fixed compact set of $a\in\C^{m}$,
\begin{equation}\label{estim}
\int_0^{+\infty}\left|s_{jk}(x,\lambda)\right|\,dx=O\left(\lambda^{-\rho}\right),\quad j, k = 1, 2,
\end{equation}
as $\lambda\to\infty$ in the sector $|\arg(\lambda)|\leq \pi-\delta$, where $\rho=\frac{1}{2}+\frac{1}{m}$.

Then from \eqref{asy_res}, we  obtain that $\sup_{x\geq0}|W_1(x,\lambda)|\leq 1$ and
$\sup_{x\geq0}|W_2(x,\lambda)|\leq 2$ for all large enough $\lambda$ in the sector $|\arg(\lambda)|\leq
\pi-\delta$. Thus, from \eqref{int-1} and \eqref{int-2}, and the fact that $\Re\left(\sqrt{Q(x)}\right)\geq 0$ for
all $x\geq0$ for all large $\lambda$ in the sector $|\arg(\lambda)|\leq\pi-\delta$, one can show that
\begin{align}
|W_1(0,\lambda)|&\leq\int_0^{+\infty}\left(|s_{11}(t,\lambda)|+2|s_{12}(t,\lambda)|\right)\,dt=O\left(\lambda^{-\rho}\right),\nonumber\\
|W_2(0,\lambda)|&\leq1+\int_0^{+\infty}\left(|s_{21}(t,\lambda)|+2|s_{22}(t,\lambda)|\right)\,dt=1+O\left(\lambda^{-\rho}\right),\nonumber
\end{align}
where we used Lemma \ref{estim}.  Hence, these estimates along with \eqref{asy_w12_eq} and \eqref{fu_rel}
completes proof.
\end{proof}
\begin{remark}
{\rm We can better approximate solutions by solving integral equations \eqref{int-1} and \eqref{int-2} by
iteration, beginning with $w_1(x)=0$ and $w_2(x)=1$. Then use Lemma \ref{gen_for} and ideas in the proof of Lemma
\ref{asy_lemma}, one can improve the error terms further.}
\end{remark}

\section{\bf Asymptotics near the negative real $\lambda$-axis}
\label{prop_sect}

In this section, we will obtain asymptotics of $f(0, {a},\lambda)$ and $f^\d(0, {a},\lambda)$ as
$\lambda\to\infty$ in a sector containing the negative real $\lambda$-axis.

Let us review some additional properties of solutions of \eqref{complex-1}.
 These solutions have simple asymptotic behavior as
$z\to\infty$ in the complex $z$-plane \cite[\S 7.4]{Hille}. In order to describe this simple asymptotic behavior,
we introduce the Stokes sectors.
\begin{definition}
{\rm The {\it  Stokes sectors} $S_k$ of the equation \eqref{complex-1} are
$$S_k=\left\{z\in \C:\left|\arg (z)-\frac{2k\pi}{m+2}\right|<\frac{\pi}{m+2}\right\}\quad\text{for}\quad k\in \Z.$$ }
\end{definition}
Hille \cite[\S 7.4]{Hille} showed that every nonconstant solution of \eqref{complex-1} either decays to zero or
blows up exponentially, in each Stokes sector $S_k$. Notice from \eqref{Si-eq1} that $f(z, {a},\lambda)$ in
\eqref{Si-eq1} decays in $S_0$ and blows up in $S_{-1}\cup S_1$.

A few notations and definitions are in order. We will use
\begin{equation}\label{G_def}
G^{\ell}( {a}):=(\omega^{-\ell}a_1, \omega^{-2\ell}a_2,\ldots,\omega^{-m\ell}a_{m}) \quad \text{for}\quad \ell\in
\Z,\,\,\text{where }\,\,\omega=\exp\left[\frac{2\pi i}{m+2}\right].
\end{equation}
Then from the definition of $b_{j,k}( {a})$, for $a\in\C^{m}$ fixed,
\begin{equation}\nonumber
b_{j,k}(G^{\ell}( {a}))=\omega^{-j\ell}b_{j,k}( {a}),\quad \ell\in\Z.
\end{equation}
Also one can check straightforwardly that
\begin{equation}\label{fk_def}
f_k(z, {a},\lambda):=f(\omega^{-k}z,G^k( {a}),\omega^{2k}\lambda),\quad k\in\Z,
\end{equation}
is a solution of \eqref{complex-1}, decays in $S_k$, and blows up in $S_{k-1}\cup S_{k+1}$. In particular,
$f_{0}(z, {a},\lambda)$ and $f_{-1}(z, {a},\lambda)$ are linearly independent and
\begin{equation}\label{basic_eq}
f_{1}(z, {a},\lambda)=C( {a},\lambda)f_{0}(z, {a},\lambda)+\widetilde{C}( {a},\lambda)f_{-1}(z,
{a},\lambda)\quad\text{for some $C( {a},\lambda)$ and $\widetilde{C}( {a},\lambda)$.}
\end{equation}
Then one can see that these coefficients can be expressed as
\begin{equation}\label{C_def}
C( {a},\lambda)=\frac{\mathcal{W}_{-1,1}( {a},\lambda)}{\mathcal{W}_{-1,0}( {a},\lambda)}\quad\text{and} \quad
\widetilde{C}( {a},\lambda)=\frac{\mathcal{W}_{1,0}( {a},\lambda)}{\mathcal{W}_{-1,0}( {a},\lambda)},
\end{equation}
where $\mathcal{W}_{j,\ell}=f_jf_{\ell}^\d -f_j^\d f_{\ell}$ is the Wronskian of $f_j$ and $f_{\ell}$.
 Moreover,
we have the following lemma.
\begin{lemma}\label{shift_lemma}
Suppose $k,\,j\in\Z$. Then
\begin{equation}\label{kplus1}
\mathcal{W}_{k+1,j+1}( {a},\lambda)=\omega^{-1}\mathcal{W}_{k,j}(G( {a}),\omega^2\lambda),
\end{equation}
and $\mathcal{W}_{0,1}( {a},\lambda)=2\omega^{\mu( {a})}$.
\end{lemma}
\begin{proof}
See Sibuya \cite[pages 116-118]{Sibuya}.
\end{proof}

Now we are ready for  asymptotics of $f(0, {a},\lambda)$ and $f^\d(0, {a},\lambda)$ near the negative real
$\lambda$-axis.
\begin{theorem}
Uniformly on any compact set of $a\in\C^{m}$,
\begin{align}
\mathcal{W}_{-1,1}( {a},\lambda)=\,&\,2i\exp\left[L(G^{-1}( {a}),
\omega^{-2}\lambda)+L(G^{1}( {a}),\omega^{-m}\lambda)+O\left(\lambda^{-\rho}\right)\right],\label{wron_asy}\\
f(0, {a},\lambda)=\,&\, \omega^{\frac{1}{2}+\mu(G^{-1}( {a}))}\lambda^{-\frac{1}{4}}\exp\left[-L(G^{-1}( {a}), \omega^{-2}\lambda)+O\left(\lambda^{-\rho}\right)\right]\label{f0_asy}\\
&-i\omega^{\frac{1}{2}+\mu( {a})}\lambda^{-\frac{1}{4}}\exp\left[-L(G( {a}),
\omega^{-m}\lambda)+O\left(\lambda^{-\rho}\right)\right],\nonumber\\
f^\d(0, {a},\lambda)=\,&\, \omega^{\frac{1}{2}+\mu(G^{-1}( {a}))}\lambda^{\frac{1}{4}}\exp\left[-L(G^{-1}( {a}), \omega^{-2}\lambda)+O\left(\lambda^{-\rho}\right)\right]\label{fprime_asy}\\
&+i\omega^{\frac{1}{2}+\mu( {a})}\lambda^{\frac{1}{4}}\exp\left[-L(G( {a}),
\omega^{-m}\lambda)+O\left(\lambda^{-\rho}\right)\right],\nonumber
\end{align}
as $\lambda\to\infty$ in the sector
\begin{equation}\label{sec_neg}
\pi-\frac{4\pi}{m+2}+\delta\leq\arg(\lambda)\leq\pi+\frac{4\pi}{m+2}-\delta.
\end{equation}
\end{theorem}
\begin{proof}
The Wronskians $\mathcal{W}_{j,\ell}$ here are independent of $z$. Thus,
\begin{align}
&\mathcal{W}_{-1,1}( {a},\lambda)\nonumber\\
&=f_{-1}(z,  {a}, \lambda)f_{1}^\d(z,  {a}, \lambda)-f_{-1}^\d(z,  {a}, \lambda)f_{1}(z,  {a}, \lambda)\nonumber\\
&=\omega^{-1}f(\omega z, G^{-1}( {a}), \omega^{-2}\lambda)f^\d(\omega^{-1} z, G^{1}( {a}), \omega^{2}\lambda)-\omega f^\d(\omega z, G^{-1}( {a}), \omega^{-2}\lambda)f(\omega^{-1} z, G^{1}( {a}), \omega^{2}\lambda)\nonumber\\
&=\omega^{-1}f(0, G^{-1}( {a}), \omega^{-2}\lambda)f^\d(0, G( {a}), \omega^{-m}\lambda)-\omega f^\d(0, G^{-1}(
{a}), \omega^{-2}\lambda)f(0, G( {a}), \omega^{-m}\lambda).\nonumber
\end{align}
When $\lambda$ is in the sector \eqref{sec_neg}, we see that
$\left|\arg\left(\omega^{-2}\lambda\right)\right|\leq\pi-\delta$ and
$\left|\arg\left(\omega^{-m}\lambda\right)\right|\leq\pi-\delta$. So
we use Theorem \ref{improv} to prove \eqref{wron_asy}.

To prove \eqref{f0_asy} and \eqref{fprime_asy}, we first use \eqref{basic_eq} along with \eqref{fk_def} and
\eqref{C_def}, that is,
\begin{equation}
f(\omega^{-1}z,G( {a}),\omega^2\lambda)=\frac{\mathcal{W}_{-1,1}( {a},\lambda)}{\mathcal{W}_{-1,0}(
{a},\lambda)}f(z, {a},\lambda)+\frac{\mathcal{W}_{1,0}( {a},\lambda)}{\mathcal{W}_{-1,0}( {a},\lambda)}f(\omega
z,G^{-1}( {a}),\omega^{-2}\lambda).
\end{equation}
 Then we evaluate this and its differentiated form (w.r.t the $z$-variable) at $z=0$, solve the
resulting equations for $f(0, {a},\lambda)$ and $f^\d(0, {a},\lambda)$, and finally use Theorem \ref{improv} and
\eqref{wron_asy} to complete the proofs.

For example,
\begin{align}
f^\d(0, {a},\lambda)=&\frac{\omega^{-1}\mathcal{W}_{-1,0}( {a},\lambda)f^\d(0,G( {a}),\omega^2\lambda)-{\omega \mathcal{W}_{1,0}( {a},\lambda)}f^\d(0,G^{-1}( {a}),\omega^{-2}\lambda)}{\mathcal{W}_{-1,1}( {a},\lambda)}\nonumber\\
=&\frac{2\omega^{\mu(G^{-1}( {a}))}f^\d(0,G( {a}),\omega^{-m}\lambda)+2\omega^{1+\mu( {a})}f^\d(0,G^{-1}(
{a}),\omega^{-2}\lambda)}{2i\exp\left[L(G^{-1}( {a}), \omega^{-2}\lambda)+L(G( {a}),
\omega^{-m}\lambda)+O\left(\lambda^{-\rho}\right)\right]}\nonumber
\end{align}
as $\lambda\to\infty$ in the sector \eqref{sec_neg}. Then we use
Theorem \ref{improv} to complete the proof of \eqref{f0_asy}.
\end{proof}


\section{\bf Proof of Theorem \ref{the_main}}\label{main_thm_sec}
In this section, we prove Theorem \ref{the_main}.

First, we use the asymptotics of $f(0,  {a}, \lambda)$ and $f^\d(0,  {a}, \lambda)$ in Theorem \ref{improv} into
$\cos\theta f(0,  {a}, \lambda)+\sin\theta f^\d(0,
 {a}, \lambda)=0$ and rearrange the resulting asymptotic equation to get
\begin{equation}\label{asy_eq}
i\omega^{\mu(G^{-1}( {a}))-\mu(
{a})}=\frac{\sin\theta-\cos\theta\lambda^{-\frac{1}{2}}}{\sin\theta+\cos\theta\lambda^{-\frac{1}{2}}} \exp\left[H(
{a},\lambda)+O\left(\lambda^{-\rho}\right)\right],
\end{equation}
where
\begin{equation}
H( {a},\lambda)=L(G^{-1}( {a}), \omega^{-2}\lambda)-L(G^{1}( {a}), \omega^{-m}\lambda).
\end{equation}
So when $\sin\theta\not=0$, since
$\frac{1-t}{1+t}=1-2t+O\left(t^2\right)=\exp\left[-2t+O\left(t^2\right)\right]$
as $t\to0$, we have
\begin{equation}
i\exp\left[\frac{4\nu( {a})\pi i}{m+2}\right]=\exp\left[H(
{a},\lambda)-\frac{2\cot\theta}{\lambda^{\frac{1}{2}}}+O\left(\lambda^{-\rho}\right)\right].
\end{equation}

Then since, from Lemma \ref{asy_lemma},
\begin{align}
H( {a},\lambda)-\frac{2\cot\theta}{\lambda^{\frac{1}{2}}}+O\left(\lambda^{-\rho}\right)&=K_{m}\left(\omega^{-2}\lambda\right)^{\frac{1}{2}+\frac{1}{m}}(1+o(1))-K_{m}\left(\omega^{-m}\lambda\right)^{\frac{1}{2}+\frac{1}{m}}(1+o(1))\nonumber\\
&=K_m\left(\exp\left[-\frac{2\pi}{m} i\right]-\exp\left[-\pi i\right]\right)\lambda^{\frac{1}{2}+\frac{1}{m}}(1+o(1))\nonumber\\
&=K_m\left(1+\exp\left[-\frac{2\pi}{m}
i\right]\right)\lambda^{\frac{1}{2}+\frac{1}{m}}(1+o(1)),\nonumber
\end{align}
and since $\arg\left(1+\exp\left[-\frac{2\pi}{m}
i\right]\right)=-\frac{\pi}{m}$, we have
\begin{equation}\nonumber
\arg\left(H(
{a},\lambda)-\frac{2\cot\theta}{\lambda^{\frac{1}{2}}}+O\left(\lambda^{-\rho}\right)\right)=-\frac{\pi}{m}+\frac{m+2}{2m}\arg(\lambda)+o(1).
\end{equation}
Thus, if $\lambda$ lies in \eqref{sec_neg} and $|\lambda|$ is large,
we have
\begin{equation}
\left|\arg\left(H(
{a},\lambda)-\frac{2\cot\theta}{\lambda^{\frac{1}{2}}}+O\left(\lambda^{-\rho}\right)\right)-\frac{\pi}{2}\right|\leq\frac{2\pi}{m}+o(1).
\end{equation}

So $\lambda\mapsto H(
{a},\lambda)-\frac{2\cot\theta}{\lambda^{\frac{1}{2}}}+O\left(\lambda^{-\rho}\right)$
maps the sector \eqref{sec_neg} near infinity onto a region
containing $|\arg(\lambda)-\frac{\pi}{2}|\leq \varepsilon_1$ and
$|\lambda|\geq M_0$ for some positive real numbers $\varepsilon_1,
M_0.$ Hence, there exists a sequence of the numbers $\lambda_n$ in
\eqref{sec_neg} for each  $n\geq N_0$ for some positive integer
$N_0=N_0( {a},\theta)$ such that $ \cos\theta f(0,  {a},
\lambda_n)+\sin\theta f^\d(0,  {a}, \lambda_n)=0 $ and hence
\begin{equation}\label{bas_eqn}
H_1( {a},\lambda_n):=H( {a},\lambda_n)-\frac{2\cot\theta}{\lambda_n^{\frac{1}{2}}}-\frac{4\nu( {a})\pi
i}{m+2}+O\left(\lambda_n^{-\rho}\right)\underset{n\to+\infty}{=}2n\pi i+\frac{\pi}{2}i.
\end{equation}
Next, by Lemma \ref{asy_lemma},
\begin{align}
H( {a},\lambda_n)=&L(G^{-1}( {a}),\omega^{-2}\lambda_n)-L(G( {a}),\omega^{-m}\lambda_n)\nonumber\\
=&\sum_{j=0}^{\infty}\left(K_{m,j}(G^{-1}( {a}))(\omega^{-2}\lambda_n)^{\frac{1}{2}+\frac{1-j}{m}}-K_{m,j}(G( {a}))(\omega^{-m}\lambda_n)^{\frac{1}{2}+\frac{1-j}{m}}\right)\nonumber\\
&-\frac{\nu(G^{-1}( {a}))}{m}\ln(\omega^{-2}\lambda_n)+\frac{\nu(G( {a}))}{m}\ln(\omega^{-m}\lambda_n)\nonumber\\
=&\sum_{j=0}^{\infty}\left(\omega^{j}\omega^{-2(\frac{1}{2}+\frac{1-j}{m})}-\omega^{-j}\omega^{-m(\frac{1}{2}+\frac{1-j}{m})}\right)K_{m,j}( {a})\lambda_n^{\frac{1}{2}+\frac{1-j}{m}}+\frac{\nu( {a})}{m}\frac{(2m-4)\pi}{m+2}i\nonumber\\
=&2i\sum_{j=0}^{\infty}\sin\left(\frac{(m-2+2j)\pi}{2m}\right)K_{m,j}(
{a})\left(-\lambda_n\right)^{\frac{1}{2}+\frac{1-j}{m}}+\frac{\nu( {a})}{m}\frac{(2m-4)\pi}{m+2}i.\nonumber
\end{align}
So from this and \eqref{bas_eqn} one can get that
\begin{equation}\label{the_asy_eqn}
\frac{1}{\pi}\sum_{j=0}^{m+1}\cos\left(\frac{(j-1)\pi}{m}\right)K_{m,j}(
{a})\left(-\lambda_n\right)^{\frac{1}{2}+\frac{1-j}{m}}-\frac{\nu( {a})}{m} +\frac{\cot\theta}{\pi
\left(-\lambda_n\right)^{\frac{1}{2}}}+O\left(\lambda_n^{-\rho}\right)\underset{n\to+\infty}{=}\left(n+\frac{1}{4}\right),
\end{equation}
where we used $\sin\left(\frac{\pi}{2}+t\right)=\cos t$.

So far, we have showed that there exists $N_0>0$ such that if $n\geq N_0$, then there exists $\lambda_n$ such that
\eqref{the_asy_eqn} holds. There are two issues that we need to examine further in the above argument. In
\eqref{bas_eqn} we did not discuss whether or not there is exactly one $\lambda_n$ for each and every $n\geq N_0$.
Also, we did not examine how many eigenvalues are not in the set $\{-\lambda_n\}_{n\geq N_0}$. These issues are
important because different staring points of the index set  create  an $O(1)$ error.

In order to resolve these issues, we will use three ingredients. First, we will use that when the potential is
$V(x)=x^m$, there is exactly one $-\lambda_n$ for each and every $n\geq N_0$. Moreover, there are exactly $N_0$
eigenvalues that are not in $\{-\lambda_n\}_{n\geq N_0}$ \cite{TIT1}. Second, we will show and use that for all
large $n$, $|\lambda_n|<|\lambda_{n+1}|$. Third, we will show and use  continuity of eigenvalues.

{\it Step 1}: First, we will use the following result of Titchmarsh \cite{TIT1}.  The eigenvalue
problem $-y^\dd+x^m y=Ey$ under the boundary condition \eqref{bd_cond} has infinitely many eigenvalues
$\{E_n\}_{n\geq 0}$ with $E_n<E_{n+1}$ for all $n\geq 0$ such that
\begin{equation}\label{Titch}
\frac{1}{\pi}\int_0^{E_n^{\frac{1}{m}}}\sqrt{E_n-x^m}\,dx\underset{n\to+\infty}{=}\left(n+\frac{1}{4}\right)+o\left(1\right),
\end{equation}
provided that $\sin\theta\not=0$ (see, e.g., \cite{TIT2,TIT1} and references therein). Of course, we can order the
eigenvalues $E_n$ in this case in the order of their magnitudes. For the potential $V(x)=x^m$,
$H(0,\lambda_n)=(2\pi)^{-1}B\left(\frac{1}{2},1+\frac{1}{m}\right)\left(-\lambda_n\right)^{\rho}$ and the left-hand side of \eqref{Titch} is
$(2\pi)^{-1}B\left(\frac{1}{2},1+\frac{1}{m}\right)E_n^{\rho}$.  Also, the error terms in both \eqref{the_asy_eqn} and \eqref{Titch} are in $o(1)$. Thus, we
have $E_n=-\lambda_n$ for all $n\geq N_0$ and we can number $\lambda_n$ for $0\leq n<N_0$ so that
$E_n=-\lambda_n$. In other words, there exists exactly one $\lambda_n$ for each and every $n\geq N_0$ in
\eqref{bas_eqn} and there are exactly $N_0$ eigenvalues that are not in the set $\{-\lambda_n\}_{n\geq N_0}$ so
that we can label the eigenvalues $E_n=-\lambda_n$ for $n\geq 0$.

{\it Step 2}: Next, we will show monotonicity of magnitudes of large eigenvalues for each $ {a}\in\C$.
Let us use $\lambda_n( {a},\theta)$ instead of $\lambda_n$ to indicate their dependence on $a$ and $\theta$. We
will show
\begin{equation}\label{mono1}
|\lambda_n( {a},\theta)|<|\lambda_{n+1}( {a},\theta)|\,\,\text{ for all large $n$.}
\end{equation}
Here we allow the possibility of having more than one $\lambda_n( {a},\theta)$ for some $n\geq N_0$ and
\eqref{mono1} holds for any choice of such $\lambda( {a},\theta)$.

Since \eqref{bas_eqn} is an asymptotic formula, one can solve it for $\lambda_n( {a},\theta)$ and obtain
asymptotics of $-\lambda_n$ in terms of $n$. That is, there are $d_j( {a},\theta)$, $0\leq j\leq m+1$ such that
\begin{equation}\label{mono}
-\lambda_n=\sum_{j=0}^{m+1}d_j(
{a},\theta)\left(n+\frac{1}{4}\right)^{\frac{2m}{m+2}\left(1-\frac{j}{m}\right)}+O\left(n^{-\frac{4}{m+2}}\right),
\end{equation}
as $n\to\infty$, where $d_0( {a},\theta)>0$.
\begin{equation}\label{d_def}
d_0( {a},
\theta)=\frac{\pi}{K_{m,0}\cos\left(\frac{\pi}{m}\right)}=\left(\frac{2\pi}{B\left(\frac{1}{2},1+\frac{1}{m}\right)}\right)^{\frac{2m}{m+2}}.
\end{equation}
 In particular, $\arg(-\lambda_n)\to 0$. At this point, we allow
possibility of getting more than one $\lambda_n$ for some $n\geq N_0$ in \eqref{the_asy_eqn}. Equation
\eqref{mono} holds for any such $\lambda_n$. Then since
$$\left((n+1)+\frac{1}{4}\right)^{\frac{2m}{m+2}\left(1-\frac{j}{m}\right)}=\left(n+\frac{1}{4}\right)^{\frac{2m}{m+2}\left(1-\frac{j}{m}\right)}\left(1+\frac{1}{n+\frac{1}{4}}\right)^{\frac{2m}{m+2}\left(1-\frac{j}{m}\right)},$$
we obtain
\begin{equation}\label{mono2}
\text{$-\lambda_{n+1}=-\lambda_n+d\left(n+\frac{1}{4}\right)^{\frac{m-2}{m+2}}+o\left(\left(n+\frac{1}{4}\right)^{\frac{m-2}{m+2}}\right)
$ for some $d>0$.}
\end{equation}
Since $\arg(-\lambda_n)\to 0$, \eqref{mono2} implies \eqref{mono1}.

{\it Step 3}: Now we will prove some form of continuity of eigenvalues. Hurwitz's theorem  in
complex analysis states that if a sequence of analytic functions $\psi_n$ in a common domain converges uniformly
to an analytic function $\psi$ on all compact sets in the domain, then for a given open disk on the boundary of
which the limit function $\psi$ does not have any zeros, there exists $N$ such that if $n\geq N$ then $\psi_n$ and
$\psi$ have the same number of zeros in the open disk.

We fix $\theta$ and let $\phi( {a},E):=\cos\theta f(0, {a}, -E)+\sin\theta f^\d(0,  {a}, -E)$. Then $\phi( {a},E)$
is an entire function of $ {a}$ and $E$ and on any compact set of $E$, $\phi( {a},E)$ converges uniformly to $\phi(
{a}^*,E)$ as $ {a}\to {a}^*$. Hence, for each $ {a}^*$ and a given open disk whose boundary does not contain any
zeros of $E\mapsto \phi( {a}^*,E)$, there exists $\delta>0$ such that if $| {a}-{a}^*|<\delta$, both $\phi( {a},E)$ and $\phi( {a}^*,E)$ have the same
number of zeros inside the open disk by Hurwitz's theorem.

{\it Step 4}: Now we are ready to prove that for each $n$ large,
there exists only one $\lambda_n$ that satisfies \eqref{bas_eqn}.

Since $W_1(0,a,\lambda)$ and $W_2(0,a,\lambda)$ in \eqref{asy_w12_eq} are continuous functions of $a$ and $\lambda$,
$H_1(a,\lambda)$ in \eqref{bas_eqn} is a continuous function of $a$ and $\lambda$. Thus, $H_1(a,\lambda_n(a))$  is
a continuous function of $a$ and  \eqref{bas_eqn} holds for all $n\geq N_0$. Since for the potential $V(x)=x^m$,
eigenvalues $E_n(0)=-\lambda_n(0)$ are all simple, we conclude that for $n\geq N_0$, there exists exactly
one $\lambda_n(a)$ satisfying \eqref{bas_eqn}.

Then, by Step 3, one sees that there are exactly $N_0$-eigenvalues that are not in $\{-\lambda_n(a)\}_{n\geq N_0}$
since that is the case for $a=0$. Since we discuss the asymptotics of eigenvalues  in this paper, how we order
these finitely many eigenvalues will not change the results.

Next we define
\begin{equation}\label{c_def}
c_j( {a})=\left\{
\begin{array}{ll}
 &\frac{1}{\pi}\cos\left(\frac{(j-1)\pi}{m}\right)K_{m,j}( {a})\quad\text{if $m$ is odd or $j\not=\frac{m}{2}+1$ when
$m$ is even,} \\
 &-\frac{\nu( {a})}{m}\quad\text{if $m$ is even and $j=\frac{m}{2}+1$,}
\end{array}
\right.
\end{equation}
where $K_{m,j}( {a})$ will be defined in the appendix. This completes the proof of \eqref{N_boundary}.

Similarly, when $\sin\theta=0$,  from \eqref{asy_eq} we get
\begin{equation}
-i\exp\left[\frac{4\nu( {a})\pi i}{m+2}\right]=\exp\left[H( {a},\lambda)+O\left(\lambda^{-\rho}\right)\right]
\end{equation}
and one can prove \eqref{D_boundary} for which we use \eqref{Titch} with $\left(n+\frac{1}{4}\right)$ replaced by
$\left(n+\frac{3}{4}\right)$ \cite[\S 11]{TIT1}.

\begin{remark}{\rm
Three properties of $c_j( {a})$ are used   in the proof of Corollary \ref{cor5}. We will see below that these
three properties are due to similar properties of  $b_{j,k}( {a})$ by \eqref{Kmj=}.

Equation \eqref{Kmj=} reads
\begin{equation}
K_{m,j}( {a})=\sum_{k={\lfloor\frac{j-1}{m}\rfloor}+1}^{j}K_{m,j,k}b_{j,k}( {a}),
\end{equation}
and
\begin{equation}
\text{$b_{j,k}( {a})$ is the coefficient of $z^{mk-j}$ in ${\frac{1}{2}\choose{k}}\left(P(z)\right)^k$\,\,for
$1\leq k\leq j$.}
\end{equation}
So if $1\leq j\leq m$, then
\begin{equation}
K_{m,j}( {a})=\frac{K_{m,j,1}}{2}\,a_j+\sum_{k=2}^{j}K_{m,j,k}b_{j,k}( {a})
\end{equation}
that is a non-constant linear function of $a_j$. Note also that
$K_{m,j,k}$ are absolute constants and for $2\leq k\leq j\leq m$,
$b_{j,k}( {a})$ are polynomials in $a_1,\,a_2,\dots, a_{j-1}$ and
independent of $a_j$. Also when $m$ is odd, $\nu( {a})=0$ and when
$m$ is even, $\nu( {a})=b_{\frac{m}{2}+1}( {a})$ is a non-constant
linear function of $a_{\frac{m}{2}+1}$ from the definition. }
\end{remark}

\subsection*{{\bf Acknowledgments}}
The author thanks Richard Laugesen for reading a part of this
manuscript and suggestions for improving its presentation.

\appendix
\section{\bf Computing $K_{m,j,k}$} \label{sA}
\renewcommand{\theequation}{A.\arabic{equation}}
\renewcommand{\thetheorem}{A.\arabic{theorem}}
\setcounter{theorem}{0} \setcounter{equation}{0}

In this appendix we will give a complete description of the term $L( {a},\lambda)$. We begin by defining some
constants needed for $L( {a},\lambda)$.

\begin{lemma}\label{A1}
Let $m\geq 3$ be an integer. Then
\begin{equation}\nonumber
K_m=K_{m,0}:=\int_0^{\infty}\left(\sqrt{1+t^m}-t^{\frac{m}{2}}\right)\,dt=\frac{B\left(\frac{1}{2},1+\frac{1}{m}\right)}{2\cos\left(\frac{\pi}{m}\right)}.
\end{equation}
\end{lemma}
\begin{proof}
Set $\sqrt{u}=\sqrt{1+t^m}-t^{\frac{m}{2}}$. Then
$u=1+2t^m-2t^{\frac{m}{2}}\sqrt{1+t^m}$ and
$\frac{1-u}{2\sqrt{u}}=t^{\frac{m}{2}}$. Thus
\begin{align}
\int_0^{\infty}\left(\sqrt{1+t^m}-t^{\frac{m}{2}}\right)\,dt&=\frac{1}{2^{\frac{2}{m}}m}\int_0^1\left((1-u)^{\frac{2}{m}-1}u^{\frac{1}{2}-\frac{1}{m}-1}+(1-u)^{\frac{2}{m}-1}u^{\frac{3}{2}-\frac{1}{m}-1}\right)\,du\nonumber\\
&=\frac{1}{2^{\frac{2}{m}}m}\left(B\left(\frac{2}{m},\frac{1}{2}-\frac{1}{m}\right)+B\left(\frac{2}{m},\frac{3}{2}-\frac{1}{m}\right)\right),\nonumber
\end{align}
where $B(z,w)$ is the beta function. Then we use the following to
complete the proof.
\begin{align}
&\Gamma(z+1)=z\Gamma(z),\quad \Gamma(z)\Gamma(1-z)=-z\Gamma(-z)\Gamma(z)=\frac{\pi}{\sin(\pi z)}\nonumber\\
&B(z,w)=\int_0^1(1-u)^{z-1}u^{w-1}\,du=\frac{\Gamma(z)\Gamma(w)}{\Gamma(z+w)},\quad
\Gamma(2z)=\frac{2^{2z-\frac{1}{2}}}{\sqrt{2\pi}}\Gamma(z)\Gamma\left(z+\frac{1}{2}\right).\label{gamma_eq}
\end{align}
\end{proof}

\begin{lemma}
Let $m\geq 3$ and $1\leq k\leq j\leq \frac{m+2}{2}$. Then
\begin{align}\nonumber
K_{m,j,k}&:=\int_0^{\infty}\left(\frac{t^{mk-j}}{(1+t^m)^{k-\frac{1}{2}}}-t^{\frac{m}{2}-j}\right)\,dt\nonumber\\
&=\left\{
                    \begin{array}{cl}
-\frac{2}{m}
\quad &\text{if $j=k=1$},\\
&\\
                  \frac{1}{m}B\left(k-\frac{j-1}{m},\,\frac{j-1}{m}-\frac{1}{2}\right) \quad &\text{if $1\leq k\leq j<\frac{m+2}{2}$, $j\not=1$},\\
&\\
                 \frac{2}{m}\left(\ln 2-\frac{1}{1}-\frac{1}{3}-\dots-\frac{1}{2k-5}-\frac{1}{2k-3}\right) \quad &\text{if $1\leq k\leq j=\frac{m+2}{2}$, $m$ even.}
                    \end{array}\right.\nonumber
\end{align}
\end{lemma}
\begin{proof}
The case when $j=k=1$ is an easy consequence of
$$\frac{d}{dt}\left(\sqrt{1+t^m}-t^{\frac{m}{2}}\right)=\frac{m}{2}\left(\frac{t^{m-1}}{(1+t^m)^{\frac{1}{2}}}-t^{\frac{m}{2}-1}\right).$$

Suppose that $1\leq k\leq j\leq\frac{m+1}{2}$, $j\not=1$. Then since
\begin{align}
&\frac{d}{dt}\left(\frac{t^{mk-(j-1)}}{(1+t^m)^{k-\frac{1}{2}}}-t^{\frac{m}{2}-(j-1)}\right)\nonumber\\
&=(mk-(j-1))\left(\frac{t^{mk-j}}{(1+t^m)^{k-\frac{1}{2}}}-t^{\frac{m}{2}-j}\right)-m(k-\frac{1}{2})\left(\frac{t^{m(k+1)-j}}{(1+t^m)^{(k+1)-\frac{1}{2}}}-t^{\frac{m}{2}-j}\right),\nonumber
\end{align}
we have
\begin{align}
\int_0^{\infty}\left(\frac{t^{mk-j}}{(1+t^m)^{k-\frac{1}{2}}}-t^{\frac{m}{2}-j}\right)dt
&=\frac{m(k-1)-(j-1)}{m(k-1)-\frac{m}{2}}\int_0^{\infty}\left(\frac{t^{m(k-1)-j}}{(1+t^m)^{(k-1)-\frac{1}{2}}}-t^{\frac{m}{2}-j}\right)dt\nonumber\\
&=\frac{\Gamma\left(k-\frac{j-1}{m}\right)\Gamma\left(1-\frac{1}{2}\right)}{\Gamma\left(k-\frac{1}{2}\right)\Gamma\left(1-\frac{j-1}{m}\right)}\int_0^{\infty}\left(\frac{t^{m-j}}{(1+t^m)^{\frac{1}{2}}}-t^{\frac{m}{2}-j}\right)dt.\nonumber
\end{align}
Next, we use the substitution
$\sqrt{u}=\sqrt{1+t^m}-t^{\frac{m}{2}}$ to show
$$\int_0^{\infty}\left(\frac{t^{m-j}}{(1+t^m)^{\frac{1}{2}}}-t^{\frac{m}{2}-j}\right)dt=-\frac{2^{\frac{2(j-1)}{m}}}{m}B\left(1-\frac{2(j-1)}{m},\frac{1}{2}+\frac{(j-1)}{m}\right).$$
Finally, we use equations in \eqref{gamma_eq} to complete the proof
for $1\leq k\leq j\leq\frac{m+1}{2}$, $j\not=1$.

\end{proof}

\begin{lemma}
Let $m\geq 4$ be an even integer and let $1\leq k\leq j= \frac{m+2}{2}$. Then
\begin{equation}\nonumber
K_{m,j,k}:=\int_0^{\infty}\left(\frac{t^{mk-j}}{(1+t^m)^{k-\frac{1}{2}}}-\frac{1}{1+t}\right)\,dt=\frac{2}{m}\left(\ln 2-\frac{1}{1}-\frac{1}{3}-\dots-\frac{1}{2k-5}-\frac{1}{2k-3}\right).
\end{equation}
\end{lemma}
\begin{proof}
By integration
by parts, for $R>0$,
\begin{equation}\nonumber
\int_0^{R}\frac{t^{mk-\frac{m}{2}-1}}{\left(1+t^m\right)^{k-\frac{1}{2}}}dt=\left.\frac{1}{m\left(-k+\frac{3}{2}\right)}\frac{t^{m(k-1)-\frac{m}{2}}}{\left(1+t^m\right)^{(k-1)-\frac{1}{2}}}\right|_0^R+\int_0^Rt^{m(k-1)-\frac{m}{2}-1}\frac{1}{\left(1+t^m\right)^{(k-1)-\frac{1}{2}}}dt,
\end{equation}
and hence
$$\int_0^{\infty}\left(\frac{t^{mk-\frac{m}{2}-1}}{\left(1+t^m\right)^{k-\frac{1}{2}}}-\frac{t^{m(k-1)-\frac{m}{2}-1}}{\left(1+t^m\right)^{(k-1)-\frac{1}{2}}}\right)\,dt=-\frac{2}{m\left(2k-3\right)}.$$
Also, since
$$
\int\left(\frac{t^{\frac{m}{2}-1}}{\left(t^m+1\right)^{\frac{1}{2}}}-\frac{1}{t+1}\right)dt=\frac{2}{m}\ln\left(\sqrt{1+t^m}+t^{\frac{m}{2}}\right)-\ln(1+t)+C,
$$
one sees that
$$\int_0^{\infty}\left(\frac{t^{\frac{m}{2}-1}}{\left(1+t^m\right)^{\frac{1}{2}}}-\frac{1}{1+t}\right)\,dt=\frac{2\ln 2}{m}.
$$
\end{proof}

\begin{lemma}\label{gen_for}
Let $m\geq 3$ and $j\geq \frac{m+3}{2}$. Let $\ell\geq 0$ be  such that   $\ell m\leq j< (\ell+1)m$. If
$\ell+1\leq k\leq j<(\ell+1)m$, then
\begin{equation}\nonumber
K_{m,j,k}:=\int_0^{\infty}\frac{t^{mk-j}}{(1+t^m)^{k-\frac{1}{2}}}\,dt=\frac{1}{m}B\left(k-\frac{j-1}{m},\frac{j-1}{m}-\frac{1}{2}\right).
\end{equation}
\end{lemma}
\begin{proof}
Use the change of the variable $1+t^m=\frac{1}{u}$ and the
definition of the beta function.
\end{proof}

Next we will give a complete description of $L( {a},\lambda)$. First, we define for $\ell m\leq j\leq (\ell+1)m-1$
with $\ell=0, 1, 2, \dots$,
\begin{equation}\nonumber
g_j(\tau)=\sum_{k=\ell+1}^{j}\frac{b_{j,k}( {a})\tau^{mk-j}}{\left(\tau^m+1\right)^{k-\frac{1}{2}}}.
\end{equation}
Then $\sum_{k=\ell+1}^jb_{j,k}( {a})=b_j( {a})$.

\begin{lemma}\label{asy_lemma}
Let $m\geq 3$ and $a\in\C^{m}$ be fixed. Then there exist
  $K_{m,j}( {a})\in\C$ such
that
\begin{equation}
L( {a},\lambda)=\left\{\begin{array}{rl}
&\sum_{j=0}^{\infty}K_{m,j}( {a})\lambda^{\frac{1}{2}+\frac{1-j}{m}}\,\quad\text{if $m$ is odd,}\\
&\sum_{j=0}^{\infty}K_{m,j}( {a})\lambda^{\frac{1}{2}+\frac{1-j}{m}}-\frac{b_{\frac{m}{2}+1}(
{a})}{m}\ln\lambda\,\quad\text{if $m$ is even,}
\end{array}
\right.\nonumber
\end{equation}
as $\lambda\to\infty$ in the sector $|\arg(\lambda)|\leq\pi-\delta$,
where
\begin{align}
 K_{m,0}( {a})&=K_{m}=\int_0^{\infty}\left(\sqrt{1+t^m}-\sqrt{t^m}\right)\, dt>0\quad\text{for all $m\geq 3$},\nonumber\\
K_{m,j}( {a})&=\int_0^{\infty}\left(g_j(t)-b_j( {a})t^{\frac{m}{2}-j}\right)\,dt\quad\text{for all $1\leq j\leq\frac{m+1}{2}$},\label{def_K}\\
K_{m,\frac{m}{2}+1}( {a})&=\int_0^{\infty}\left(g_{\frac{m}{2}+1}(t)-\frac{b_{\frac{m}{2}+1}(
{a})}{t+1}\right)\,dt\quad\text{when
$m$ is even},\nonumber\\
K_{m,j}( {a})&=\int_0^{\infty}g_j(t)\,dt\quad\text{for all $j\geq\frac{m+3}{2}$}.\label{def_K1}
\end{align}
Moreover,
\begin{equation}\label{Kmj=}
K_{m,j}( {a})=\sum_{k=\lfloor\frac{j}{m}\rfloor+1}^{j}K_{m,j,k}\,b_{j,k}( {a}).
\end{equation}
\end{lemma}
\begin{proof}
The function $L( {a},\lambda)$ is defined as an integral over $0\leq t<+\infty$. Namely,

\begin{equation}
L( {a},\lambda)=\int_0^{\infty}R(t,  {a}, \lambda)\,dt,
\end{equation}
where
\begin{align}
R(t, {a},\lambda)=\left\{
                    \begin{array}{rl}
                    &\sqrt{t^m+P_{m-1}(t)+\lambda}- t^{\frac{m}{2}}-\sum_{j=1}^{\frac{m+1}{2}}b_j( {a})t^{\frac{m}{2}-j} \quad \text{if $m$ is odd,}\\
                   &\sqrt{t^m+P_{m-1}(t)+\lambda}- t^{\frac{m}{2}}-\sum_{j=1}^{\frac{m}{2}}b_j( {a})t^{\frac{m}{2}-j}-\frac{b_{\frac{m}{2}+1}( {a})}{t+1} \quad \text{if $m$ is even.}
                    \end{array}
                         \right. \nonumber
\end{align}
 Then using Cauchy's integral formula, one can see that for all $\lambda$ with $|\lambda|$ large enough in the sector
 $|\arg(\lambda)|\leq \pi-\delta$,
\begin{equation}\label{asy_eql}
L( {a},\lambda)=\int_{0}^{+\infty}R(t, {a},\lambda)\, dt
=\lambda^{\frac{1}{m}}\int_0^{+\infty}R(\lambda^{\frac{1}{m}}\tau, {a},\lambda)\, d\tau,
\end{equation}
where
\begin{align}
R(\lambda^{\frac{1}{m}}& \tau,  {a},\lambda)\nonumber\\
=&\left\{
                    \begin{array}{rl}
                    &\lambda^{\frac{1}{2}}\left(\sqrt{\tau^m+1+\frac{P_{m-1}(\lambda^{\frac{1}{m}}\tau)}{\lambda}}- \tau^{\frac{m}{2}}-\sum_{j=1}^{\frac{m+1}{2}}b_j( {a})\frac{\tau^{\frac{m}{2}-j}}{\lambda^{\frac{j}{m}}}\right)\,\, \text{if $m$ is odd,}\\
                   &\lambda^{\frac{1}{2}}\left(\sqrt{\tau^m+1+\frac{P_{m-1}(\lambda^{\frac{1}{m}}\tau)}{\lambda}}- \tau^{\frac{m}{2}}-\sum_{j=1}^{\frac{m}{2}}b_j( {a})\frac{\tau^{\frac{m}{2}-j}}{\lambda^{\frac{j}{m}}}-\frac{\lambda^{-\frac{1}{2}}b_{\frac{m}{2}+1}( {a})}{\lambda^{\frac{1}{m}}\tau+1}\right)
 \,\, \text{if $m$ is even.}
                    \end{array}
                         \right. \nonumber
\end{align}

Next, we examine the following square root in $R(\lambda^{\frac{1}{m}}\tau,  {a},\lambda)$:
\begin{align}
\sqrt{\tau^m+1+\frac{P_{m-1}(\lambda^{\frac{1}{m}}\tau)}{\lambda}}
&=\sqrt{\tau^m+1}\sqrt{1+\frac{P_{m-1}(\lambda^{\frac{1}{m}}\tau)}{\lambda(\tau^m+1)}}\nonumber\\
&=\sqrt{\tau^m+1}\left(1+\sum_{k=1}^{\infty}{{\frac{1}{2}}\choose{k}}\left(\frac{P_{m-1}(\lambda^{\frac{1}{m}}\tau)}{\lambda(\tau^m+1)}\right)^k\right)\nonumber\\
&=\sqrt{\tau^m+1}+\sum_{j=1}^{\infty}\frac{g_j(\tau)}{\lambda^{\frac{j}{m}}},\nonumber
\end{align}
where
\begin{equation}\label{gjdef}
g_j(\tau)=\sum_{k=1}^{j}\frac{b_{j,k}( {a})\tau^{mk-j}}{\left(\tau^m+1\right)^{k-\frac{1}{2}}}.
\end{equation}
Thus,
\begin{align}
g_j(\tau)-b_j( {a})\tau^{\frac{m}{2}-j}&=\sum_{k=1}^jb_{j,k}( {a})\left(\frac{\tau^{mk-j}}{\left(\tau^m+1\right)^{k-\frac{1}{2}}}-\tau^{\frac{m}{2}-j}\right)\label{gjsepdef}\\
&\underset{\tau\to\infty}{=}\sum_{k=1}^jb_{j,k}( {a})\tau^{\frac{m}{2}-j}O\left(\frac{1}{\tau^m}\right)\nonumber\\
&\underset{\tau\to\infty}{=}O\left(\frac{1}{\tau^{\frac{m}{2}+j}}\right)\quad\text{for
all $1\leq j\leq\frac{m+1}{2}$}.\nonumber
\end{align}

So $\int_0^{\infty}\left|g_j(\tau)-b_j( {a})\tau^{\frac{m}{2}-j}\right|\,d\tau<+\infty$ for all $1\leq
j\leq\frac{m+1}{2}$. Next, when $m$ is even and $j=\frac{m}{2}+1$, we write
\begin{align}
&\int_0^{\infty}\left(g_{\frac{m}{2}+1}(\tau)-\frac{b_{\frac{m}{2}+1}( {a})}{\tau+\lambda^{-\frac{1}{m}}}\right)\, d\tau\nonumber\\
&=\int_0^{\infty}\left(g_{\frac{m}{2}+1}(\tau)-\frac{b_{\frac{m}{2}+1}( {a})}{\tau+1}\right)\, d\tau+b_{\frac{m}{2}+1}( {a})\int_0^{\infty}\left(\frac{1}{\tau+1}-\frac{1}{\tau+\lambda^{-\frac{1}{m}}}\right)\, d\tau\nonumber\\
&\overset{let}{=}K_{m,\frac{m}{2}+1}( {a})-\frac{b_{\frac{m}{2}+1}( {a})}{m}\ln(\lambda),\nonumber
\end{align}
where we take $\Im(\ln(\lambda))=\arg(\lambda)\in(-\pi,\pi)$. Also,
if $j\geq\frac{m+3}{2}$, then
$\int_0^{\infty}\left|g_j(\tau)\right|\,d\tau<+\infty$.

Finally, one can deduce \eqref{Kmj=} from \eqref{gjdef} and
\eqref{gjsepdef}.
\end{proof}

{\sc email contact:} kshin@westga.edu
\end{document}